\documentclass{article}

\usepackage{authblk}

\usepackage[T1]{fontenc}

\usepackage[utf8]{inputenc}
\usepackage[english]{babel}
\usepackage{csquotes}
\usepackage{amsmath, amssymb, amsfonts, amsthm}

\usepackage{braket}
\usepackage{graphicx}

\usepackage{xcolor}
\definecolor{URLColorBlue}{HTML}{2A1B81}
\definecolor{BlueGreen}{cmyk}{0.85,0,0.33,0}
\definecolor{RawSienna}{HTML}{A0522D}
\definecolor{TealLight}{HTML}{00C798}
\definecolor{Teal}{HTML}{008080}

\usepackage{hyperref}

\hypersetup{colorlinks,
			citecolor=TealLight,
			urlcolor=URLColorBlue,
			pdftitle={weight-distribution bound},
			pdfauthor={Alessio Meneghetti, Massimiliano Sala, Alex Tomasi},
			pdfkeywords={Random number generation} {linear binay codes}{entropy extractor},
}

\theoremstyle{plain}% default
\newtheorem{thm}{Theorem}

\newtheorem{cor}{Corollary}

\theoremstyle{definition}
\newtheorem{defn}{Definition}

\theoremstyle{remark}

\title{A weight-distribution bound for entropy extractors using linear binary codes}
\author{A. Meneghetti}
\author{M. Sala}
\author{A. Tomasi}
\affil{\small{Department of Mathematics, University of Trento}}

\begin{document}

\maketitle

\begin{abstract}
	We consider a bound on the bias reduction of a random number generator by processing based on binary linear codes. We introduce a new bound on the total variation distance of the processed output based on the weight distribution of the code generated by the chosen binary matrix. Starting from this result we show a lower bound for the entropy rate of the output of linear binary extractors.
\end{abstract}

\section{introduction}

A post-processing or entropy extractor is a deterministic algorithm able to decrease the statistical imperfections of a random number generator (RNG). A first and well-known example is the von Neumann procedure \cite{CGC-misc-art-von1951various, CGC-misc-art-peres1992iterating}, which performs best when the generator provides IID binary output. This procedure has been studied in detail in the case of non-identically distributed quantum sources \cite{CGC-misc-art-abbott2012neumann}, but its behaviour with non-independent generators is still an open problem. 
The von Neumann procedure also has a significant cost in terms of bit consumption, as well as a variable output rate. Its output bit strings have an average length of a quarter of the original sequence.

Other processing algorithms have included cryptographic primitives to randomize the output of a generator (e.g. AES), or software-based RNGs whose output is to be XORed with the output of the generator itself. Actually NIST recommendations include summing the output of a pseudo-RNG to that of each physical RNG \cite{CGC-msc-book-barker2012recommendation,CGC-misc-art-barker2012recommendation}. While the observed quality of the randomness thus produced is very good a posteriori, there is no formal proof of it a priori, meaning we cannot give a theorical description of how good we can expect the output to be.

Of particular interest for our purposes is an algebraic technique presented in \cite{CGC-misc-art-lacharme2008analysis}, where bit sequences are considered as vectors and used as input to a linear map between binary vector spaces. The map is hence defined as the left-multiplication with a binary matrix; the ratio between the number of rows and the number of columns specifies the compression ratio of this processing, and a theoretical bound on the quality of the output is provided. It is also shown that this result is strictly related to the minimum distance of the binary linear code generated by the matrix. In this work and in \cite{CGC-misc-art-lacharme2009analysis} Lacharme starts from this simple processing to describe the utilization of more general Boolean functions to build extractors.

Following the reasoning in \cite{CGC-misc-art-barker2012recommendation} we use this definition:
\begin{defn}\label{definizione generatore}
	A \textit{random number generator} (RNG) is a discrete process $Y_t$ taking values in some finite set $\Omega_y$.
	We can write $Y$ as the composition of two functions, a discrete time random process $X_t$ with values in a set $\Omega_x$ and a deterministic map $f$ from $\Omega_x$ into the set $\Omega_y$
	%\begin{equation}\nonumber
	%\begin{array}{ccc}
	\begin{align*}
		\Omega_x	& \rightarrow \Omega_y	\\
		X_t			& \mapsto f(X_t) = Y_t
	\end{align*}
	%\end{array}
	%\end{equation}
	The random process $X$ is called \textit{source of entropy} (SoE), while $f$ is a \textit{post-processing} or an \textit{(entropy) extractor}.
\end{defn}
We consider the sets $\Omega_x$ and $\Omega_y$ to be some finite fields. A binary generator is thus defined in the space $\mathbb{F}_2$, while a RNG providing sequences of $m$ bits as a single output is defined in $\mathbb{F}_{2^m}$.

The main result we present is Theorem \ref{theorem weight distribution}, which is a new bound on the output total variation distance of the binary procedure presented in \cite{CGC-misc-art-lacharme2008analysis}, namely we will provide a bound based on the entire weight distribution of the codewords generated by the chosen binary matrix, instead of relying only on the minimum distance. We are also able to show a lower bound for the entropy rate in Corollary \ref{cor entropy rate}.

\section{A bound on the total variation distance}

We consider vectors as column vectors unless otherwise specified, and we write $v^T$ for the transpose of any vector.

We consider a source of entropy $X_t$ with values in the finite field $\mathbb{F}_2$. %The first and most natural example is the case of binary sources of entropy, which correspond to $q=2$. Many theoretical results on the output of a post-processing assume that the source of entropy $X_t$ is IID. 
We can build an $n$-dimensional random variable from a binary IID SoE $X_t$ on the space $\left(\mathbb{F}_2\right)^n$, and we call this process $\bar{X}_t = (X_{t,1},\ldots,X_{t,n})^T$, $t\in\mathbb{N}$. Given a binary $k\times n$ matrix, $k<n$ we define a new IID process $\bar{Y}_t$ using the linear map

\[
	\bar{Y}_t	= \begin{pmatrix}
					Y_{t,1}\\
					\vdots \\
					Y_{t,k}
				\end{pmatrix}
				= G \cdot
				\begin{pmatrix}
					X_{t,1}\\
					\vdots \\
					X_{t,n}
				\end{pmatrix}
				= G\cdot\bar{X}_t
\]

We recall that the term \textit{bias} is usually associated to the quality of an RNG. Given a binary process $Z$, the bias is the value
$
\frac{1}{2}\varepsilon_Z = \frac{1}{2}\left| P(Z=1)-P(Z=0)\right|.
$

The bias is simply a particular case of the distance between two distributions, restricted to the binary case. A much more general measure is the following:

\begin{defn}	\label{def:TVD}
	The \textit{Total Variation Distance} between two probability measures $P$ and $Q$ over the same set $\Omega$ will be denoted with $\mathrm{TVD}(P,Q)$ and is computed as
	\[
		\mathrm{TVD}(P,Q) = \frac{1}{2}\left\| P-Q \right\|_{1} = \frac{1}{2}\sum_{\omega\in \Omega}\left| P(\omega)-Q(\omega) \right|
	\]
	Given an IID process $Z_t$ and a probability $P_Z$ on the space $\Omega_Z$, we denote its TVD from the uniform distribution $U_{\Omega_Z}$ by
	\begin{equation}\label{eq: TVD from uniform}
		\frac{1}{2}\delta_Z=\mathrm{TVD}(P_Z, U_{\Omega_Z}).
	\end{equation}
\end{defn}
The following theorem was shown in \cite{CGC-misc-art-lacharme2008analysis}, but we find it useful to provide a sketch of the proof.
\begin{thm}[Lacharme]\label{binary extractor thm}
	If the linear code generated by $G$ is an $[n,k,d]$ binary code, then the output bias is $$\frac{1}{2}\varepsilon_y \leq\frac{1}{2}\varepsilon^d.$$
\end{thm}
\begin{proof} The proof relies on observing that summing $d$ independent binary variables with a given bias yields an output with a bias reduced by an exponential factor of $d$. Multiplying a vector of such variables by $G$ as above means that each output item is the result of the sum of at least $d$ such variables.
\end{proof}
We remark that even though $\bar{Y}_t$ is an IID random variable on $\left(\mathbb{F}_2\right)^k$, it is not true that $Y_{t,i}$ has the same probability distribution for all $i=1,\ldots,k$, meaning the output of the extractor is no longer a truly binary process.
A minimal entropy bound for $\bar{Y}_t$ is provided in \cite{CGC-misc-art-lacharme2008analysis}, based on this result:
\begin{thm}[Lacharme]\label{thm infinity norm}
Given $\bar{X}$, $\varepsilon_x$, $G$ and $\bar{Y}$ as above, 
\begin{equation}\label{bound infinity norm}
P(\bar{Y}=\gamma)\leq \frac{1}{2^k}+\varepsilon_x^d
\end{equation}
for each $\gamma\in\left(\mathbb{F}_2\right)^k$
\end{thm}

In order to use Theorems \ref{binary extractor thm} and \ref{thm infinity norm} we need an IID binary source, and to have good results we need a generator matrix for a binary code with a large minimum distance. This is not easy to obtain without using large matrices, which imply a lot of computations and/or a high compression.  
Considering instead the entire weight distribution of the code we may choose smaller binary matrices.

We consider $\{X_t\}_{t>0}$ as a sequence of IID random variables with values in the finite field $\mathbb{F}_2$. Being IID, $P(X_t=\omega)=P(X_s=\omega)$ for each $t,s$ and each specific element $\omega \in \mathbb{F}_2$, i.e. we can associate a probability distribution $\rho$ to the stochastic process:
\begin{equation}
	\rho(\omega)=P(X_t=\omega)\quad \forall t
	\label{eq:rho}
\end{equation}
We denote with $U_{\mathbb{F}_2}$ the uniform distribution on the same space, so 
$$
	U_{\mathbb{F}_2}(\omega) = \frac{1}{2}	\quad\forall\omega\in\mathbb{F}_2
$$
We also use the following definition for convenience:
\begin{defn} \label{def:element-wise_variation}
	Given a random variable $Z$ with values in $\mathbb{F}_2$, we write  $P(Z=\omega)=U_{\mathbb{F}_2}+v(\omega)$. We call $v$ the \emph{element-wise variation  distance function}. 
\end{defn}
We observe $v$ has the following properties:
	\begin{enumerate}
		\item	$v:\mathbb{F}_2 \rightarrow \left[-\frac{1}{2},\frac{1}{2}\right]$
		\item	$\sum_{\omega} v(\omega)=0$
		\item	$\mathrm{TVD}(P_z,U_{\mathbb{F}_2})=\frac{1}{2}\sum_{\omega}\left|v(\omega)\right|=\frac{1}{2}\|v\|_1$
	\end{enumerate}
It may be noted by comparison with Definition \ref{def:TVD} that this is simply an element-wise measure of the total variation distance, for which we use a separate short-hand notation so as not to confuse the distance between two whole distributions with the difference in probabilities of a single element from the uniform distribution.

The process $\bar{Y}=G\cdot \bar{X}$ is now an IID process in the space $(\mathbb{F}_2)^k$ and not in $\mathbb{F}_2$. In the following theorem we provide a bound for the Total Variation Distance of the process in the space $(\mathbb{F}_2)^k$ using the weight distribution of the linear code generated by $G$.

\begin{thm}\label{theorem weight distribution}
	Let $A_l$ be the number of words with weight $l$ of the linear code generated by $G$. Let $X$ and $\bar{Y}$ be as above, and let $\frac{\delta_Y}{2}$ be the TVD of $\bar{Y}$ defined in \eqref{eq: TVD from uniform}. Then 
	\begin{equation}\label{eq our bound}
		\delta_Y\leq \sum_{l=d}^n A_l \varepsilon_x^l	
	\end{equation}
\end{thm}

\begin{proof}
First of all we remark that we can assume G is a systematic generator matrix, or can be written as such by simply applying Gaussian elimination; the application of the resulting extractor gives a process whose probability mass function is a permutation of the original probability distribution. In particular the resulting TVD is equal to the original one. We can hence consider the following extractor
\begin{equation*}
	\left(
	\begin{array}{c}
		Y_1\\
		\vdots\\
		Y_k
	\end{array}
	\right)
	=
	\begin{bmatrix}
		1&\cdots &0&g_{1,k+1}&\cdots&g_{1,n}\\
		\vdots&\ddots&\vdots&\vdots&&\vdots\\
		0&\cdots&1&g_{k,k+1}&\cdots&g_{k,n}
	\end{bmatrix}
	\cdot
	\left(
	\begin{array}{c}
		X_1\\
		\vdots\\
		X_n
	\end{array}
	\right)
\end{equation*}
If we want to compute 
\begin{equation}\label{equation single probability gamma}
P(Y_1=\gamma_1,\ldots,Y_k=\gamma_k)
\end{equation}
we can substitute to each $Y_i$ the corresponding linear combination of elements $X_1\ldots,X_n$.
We call $g_i$ the row vector $(g_{i,k+1},\ldots,g_{i,n})$, so that if $(X_{k+1},\ldots,X_n)^T$ is equal to $\omega=(\omega_{k+1},\ldots,\omega_n)^T\in\left(\mathbb{F}_2\right)^{n-k}$ we can write 
\begin{equation*}
Y_i=X_i+g_i\cdot\omega
\end{equation*}
 Applying the law of total probability, and then using the independence of the variables $X_i$ we obtain that equation \eqref{equation single probability gamma} is equal to
\begin{equation*}
\sum_{\omega\in\left(\mathbb{F}_2\right)^{n-k}}
\rho(\gamma_1+g_1\cdot\omega)
\cdots\rho(\gamma_k+g_k\cdot\omega)\cdot \rho(\omega_{k+1})\cdots \rho(\omega_{n})
\end{equation*}
We can now use the fact that $\rho(\cdot)=\frac{1}{2}+v(\cdot)$. We substitute this into the above equation and we compute the product. We call $b=(b_1,\ldots,b_n)$ a given binary vector on the space $\left(\mathbb{F}_2\right)^n$, and $W(b)$ its Hamming weight. With this notation the product can be written as
\begin{equation}\label{equation all b}
\sum_{\omega\in\left(\mathbb{F}_2\right)^{n-k}}\sum_{b\in\left(\mathbb{F}_2\right)^n}\frac{1}{2^{n-W(b)}}
v^{b_1}(\gamma_1+g_1\cdot\omega)
\cdots v^{b_n}(\omega_{n})
\end{equation}
We can exchange the two sums, and observe this value for a fixed vector $b$, i.e. we have
\begin{equation}\label{equation fixed b}
\frac{1}{2^{n-W(b)}}
\sum_{\omega\in\left(\mathbb{F}_2\right)^{n-k}}
v^{b_1}(\gamma_1+g_1\cdot\omega)
\cdots v^{b_n}(\omega_{n})
\end{equation}
We remark now that if a given $b_i=0$ then $v^{b_i}(\cdot)=1$, and if $b_i=1$ then $v^{b_i}(\cdot)=v(\cdot)$. Moreover if $W (b) = w$ then the product in \eqref{equation fixed b} is the product of exactly $w$ terms of the form $v(\cdot)\cdots v(\cdot)$ and the coefficient on the left becomes $\frac{1}{2^{n-w}}$.
\\
We introduce now the following notation:
\begin{equation*}
\left\{
\begin{array}{l}
V_1^b(\omega)=v^{b_1}(\gamma_1+g_1\cdot\omega)\cdots v^{b_k}(\gamma_k+g_k\cdot\omega)\\
V_2^b(\omega)=v^{b_{k+1}}(\omega_{k+1})\cdots v^{b_{n}}(\omega_n)
\end{array}
\right.
\end{equation*}
Equation \eqref{equation fixed b} becomes
\begin{equation}\label{equation V1 times V2}
\frac{1}{2^{n-W(b)}}
\sum_{\omega\in\left(\mathbb{F}_2\right)^{n-k}}
V_1^b(\omega)\cdot V_2^b(\omega)
\end{equation}
$V_1^b(\omega)$ is the product of some terms of the form $v(\gamma_i+g_i\cdot \omega)$. Due to the product $g_i\cdot \omega$, not necessarily all the $\omega_j$ appear as argument in all the terms. In particular, suppose that a certain $\omega_j$ appears in $V_1^b(\omega)$ an even number of times, then the map $\omega_j\mapsto V_1^b(\omega)$ is constant.
\\
To prove this observe that $V_1^b(\cdot)$ is the product of maps that can assume only two values, in fact $v(0)=\frac{\varepsilon_x}{2}=-v(1)$. If we fix $b, \gamma_1,\ldots,\gamma_k, \omega_{k+1},\ldots, \omega_n$, then $V_1^b(\omega)$ can be $\pm \left(\frac{\varepsilon_x}{2}\right)^{w_1}$, where $w_1$ is the number of terms in the product, i.e. the number of ones among the first $k$ components of $b$.
\\
If we now keep the same choices for all the parameters and we change only $\omega_j$, the maps in the product in which $\omega_j$ appears change their sign. We assumed $\omega_j$ appears in an even number of terms, hence the final result is still the same and $\omega_j\mapsto V_1^b(\omega)$ is constant.
\\
In the same way if a certain $\omega_j$ appears an odd number of times in $V_1^b(\omega)$, then the map $\omega_j\mapsto V_1^b(\omega)$ is not constant.
\\
Observe that if the first $k$ components of the vector $b$ are fixed, there is only one way to choose the last $n-k$ bits in order to obtain a codeword $c$ of the code generated by $G$. In particular, the last $n-k$ bits of $c$ have a $0$ in each position corresponding to a $\omega_j$ that appears an even number of times, and a $1$ in each position corresponing to a $\omega_j$ that appears an odd number of times.
\\
Hence we assume now that the chosen $b$ is a codeword, and we look again at equation \eqref{equation V1 times V2}. In this case the product $V_2^b(\omega)$ depends only on the elements $\omega_j$ that appears an odd number of times.
\\
If instead $b$ is not a codeword there are only two possibilities: either there is a term $v(\omega_j)$ in $V_2^b(\omega)$ with $\omega_j$ that appears an even number of times in $V_1^b$, or $v(\omega_j)$ is not in the product even though $\omega_j$ appear an odd number of times in $V_1^b$.
\\
In both cases there exists an index $j$ so that the map $\omega_j\mapsto V(\omega_j):= V_1^b(\omega)\cdot V_2^b(\omega)$ is such that $V(0)=-V(1)$. Hence
\begin{equation*}
\sum_{\omega}V_1^b(\omega)V_2^b(\omega)=\sum_{\omega_i\neq \omega_j}\left(\sum_{\omega_j}V_1^b(\omega)V_2^b(\omega)\right)=0
\end{equation*}
At this point we have that if $b$ is not a codeword, then equation \eqref{equation fixed b} is zero. Using this fact we can rewrite equation \eqref{equation all b} as
\begin{equation*}
\sum_{b\in\mathcal{C}}
\frac{1}{2^{n-W(b)}}
\sum_{\omega\in\left(\mathbb{F}_2\right)^{n-k}}
v^{b_1}(\gamma_1+g_1\cdot\omega)
\cdots v^{b_n}(\omega_{n})
\end{equation*}
where $\mathcal{C}$ is the code generated by $G$.
\\
We can now write a formula for the TVD $\frac{\delta_Y}{2}$ of $\bar{Y}$,
\begin{equation*}
\delta_Y=\sum_{\gamma\in\left(\mathbb{F}_2\right)^k}\left|P(\bar{Y}=\gamma)-\frac{1}{2^k} \right|
\end{equation*}
Using \eqref{equation V1 times V2}, we can  write $\delta_Y$ as
\begin{equation}\label{equation TVD V1 V1}
\sum_{\gamma_1,\ldots,\gamma_k}\left|
\sum_{b\in\mathcal{C}}\left(
\frac{1}{2^{n-W(b)}}
\sum_{\omega\in\left(\mathbb{F}_2\right)^{n-k}}
V_1^b(\omega)\cdot V_2^b(\omega)
\right)
-\frac{1}{2^k}\right|
\end{equation}

Consider now the word $b=0$, which belongs to $\mathcal{C}$ for every generator matrix $G$.
The term inside the parenthesis becomes $\frac{1}{2^{n}}\cdot\sum_{\omega\in\left(\mathbb{F}_2\right)^{n-k}}1=\frac{1}{2^k}$, and we can simplify it together with the term $-\frac{1}{2^k}$.
\\
Equation \eqref{equation TVD V1 V1} becomes
\begin{equation}\label{equation TVD V1 V1 no zero}
\delta_Y=\sum_{\gamma_1,\ldots,\gamma_k}\left|
\sum_{b\in\mathcal{C}\smallsetminus 0}\left(
\frac{1}{2^{n-W(b)}}
\sum_{\omega\in\left(\mathbb{F}_2\right)^{n-k}}
V_1^b(\omega)\cdot V_2^b(\omega)
\right)\right|
\end{equation}
We now apply the triangular inequality, and rearranging the terms in the sum we obtain
\begin{equation*}
\delta_Y\leq\sum_{b\in\mathcal{C}\smallsetminus 0}
\left(
\frac{1}{2^{n-W(b)}}
\sum_{\omega\in\left(\mathbb{F}_2\right)^{n-k}}
\sum_{\gamma_1,\ldots,\gamma_k}
\left|V_1^b(\omega)\cdot V_2^b(\omega)\right|
\right)
\end{equation*}
Let us now consider a codeword $b$ with weight $l$. Inside the parenthesis we get
\begin{equation}\label{equation tvd fixed b 1}
\frac{1}{2^{n-l}}
\sum_{\omega\in\left(\mathbb{F}_2\right)^{n-k}}
\left| V_2^b(\omega)\right|
\sum_{\gamma_1,\ldots,\gamma_k}
\left|V_1^b(\omega)\right|
\end{equation}
In the first $k$ components of $b$ there are $l_1\leq l$ non-zero bits, meaning $V_1^b$ is the product of $l_1$ terms, each one depending on a different $\gamma_i$, so
\begin{equation*}
\sum_{\gamma_1,\ldots,\gamma_k}
\left|V_1^b(\omega)\right|
=
2^{k-l_1}\cdot \varepsilon_x^{l_1}
\end{equation*}
This no longer depends on $\omega$, so we can group it and take it outside of the sum. Hence equation \eqref{equation tvd fixed b 1} becomes
\begin{equation*}
\frac{1}{2^{n-l}}\cdot 2^{k-l_1}\cdot \varepsilon_x^{l_1}
\sum_{\omega\in\left(\mathbb{F}_2\right)^{n-k}}
\left| V_2^b(\omega)\right|
\end{equation*}
In the same way, in the last $n-k$ components of the word $b$ there are $l-l_1$ non zero bits, so $V_2^b$ is the product of $l-l_1$ terms, each one depending on a different $\omega_i$. The sum becomes
\begin{equation}\label{equation bound fixed b}
\frac{1}{2^{n-l}}\cdot 2^{k-l_1}\cdot \varepsilon_x^{l_1}
\cdot 2^{n-k-(l-l_1)}\cdot \varepsilon_x^{l-l_1}
=
\varepsilon_x^l
\end{equation}
So, putting together equations \eqref{equation TVD V1 V1} and \eqref{equation bound fixed b}, we obtain
\begin{equation*}
\delta_Y\leq\sum_{b\in\mathcal{C}\smallsetminus0}\varepsilon_x^{W(b)}=\sum_{l=1}^{n}A_l\varepsilon_x^l
\end{equation*}

\end{proof}

This theorem provides a bound on the total variation distance of the output process considered as a vector of $k$ bits. 
\\
We remark that from Theorem \ref{thm infinity norm} we get
\begin{equation}\label{eq: worst case bound TVD}
\delta_Y\leq 2^k\left\|P(\bar{Y})-\frac{1}{2^k} \right\|_{\infty}\leq 2^k\varepsilon_x^d
\end{equation}
This is exactly \eqref{eq our bound} in the case of a code in which all the codewords have weight equal to the minimum distance. Hence, to achieve better results, among the codes with a fixed minimum distance $d$ we can look at the weight distributions in order to minimize the bound in equation \eqref{eq our bound}.

\section{An entropy bound}\label{section: entropy}
Using the bound on TVD obtained in \eqref{eq our bound} it is possible to obtain bounds on the entropy rate of the RNG. We examine one such bound to compare the effectiveness of the bound based on the weight distribution given in \eqref{eq our bound} compared to the bound based only on minimum distance, given in \eqref{eq: worst case bound TVD}. Additionally, we can then compare the bounds on Shannon entropy thus obtained to the bound on min-entropy in \cite{CGC-misc-art-lacharme2008analysis}
\begin{defn}
Given an IID process $Z_t$ taking values in $\Omega_Z$ with probability mass function $P$, the \textbf{entropy rate} is computed as 
\begin{equation}\nonumber
H(Z)=-\sum_{\omega\in\Omega_Z}P(\omega)\log_{|\Omega|}\left(P(\omega)\right)
\end{equation}
\end{defn}
We remark that $H(Z)\in[0,1]$ and that $H(Z)< 1$ whenever $P$ is not the uniform distribution. 
\\
In practice lower bounds on $H$ are used, such as the min-entropy.
\begin{defn}
The \textbf{min-entropy} $H_{\mathrm{min}}$ of $Z$ is
\begin{equation}\nonumber
H_{\mathrm{min}}(Z)=-\log_{|\Omega_Z|}\left(\max_{\omega\in\Omega_Z}P(\omega)\right)
\end{equation}
\end{defn}
Using this definition, the following lower bound on the min-entropy is shown in \cite{CGC-misc-art-lacharme2008analysis}
\begin{cor}[Lacharme]
Under the same hypotheses of Theorem \ref{thm infinity norm}
\begin{equation}\nonumber
H_{\mathrm{min}}(\bar{Y})\ge 1 - \log_{2^k}\left(1 + 2^k\cdot\varepsilon_x^{d} \right)
\end{equation}
\end{cor}
The following bound is proved in \cite{CGC-cry-prep-sason2012entropy}
\begin{thm}[Sason]\label{thm: Sason}
Let $Z_1$ and $Z_2$ be two discrete random variables that take values in a finite set $\Omega$ with probability distributions $P_1$ and $P_2$, and let $M=|\Omega|$. Then 
\begin{equation}\label{eq: bound entropy difference}
\left|H(Z_1)-H(Z_2)\right|
\leq
\frac{\delta}{2}\log_{M}(M-1)+h\left(\frac{\delta}{2}\right)
\end{equation}
where $h\left(\frac{\delta}{2}\right)=-\frac{\delta}{2}\log_{M}\delta-\left(1-\frac{\delta}{2}\right)\log_{M}\left(1-\frac{\delta}{2}\right)$, and $\frac{\delta}{2}=\mathrm{TVD}(P_1,P_2)$
\end{thm}
We apply this theorem to our case, using $P_1$ the probability distribution of the output $\bar{Y}$, $\Omega=\left(\mathbb{F}_2\right)^k$ and $P_2=U_{\Omega}$ and we obtain the following result.
\begin{cor}\label{cor entropy rate}
Under the hypotheses of Theorem \ref{theorem weight distribution} the entropy rate of $\bar{Y}$ is bounded:
\begin{equation}\nonumber
H(\bar{Y})\ge 1-\frac{\delta}{2}\log_{2^k}\left(2^k-1\right)-h\left(\frac{\delta}{2}\right)
\end{equation}
where
\begin{equation}\label{eq: def delta for H bound}
\delta=\sum_{l=d}^n A_l \varepsilon_x^l
\end{equation}
\end{cor}
The same estimate can be made using equation \eqref{eq: worst case bound TVD} rather than \eqref{eq our bound} to obtain the worst-case scenario in which all codewords are of weight $d$.
In this case \eqref{eq: def delta for H bound} becomes $$\delta=2^k\varepsilon_x^d.$$
\\
In Figures \ref{figure: RM16} and \ref{figure: RM256} we show the results of using the three lower bounds for the entropy in the case of using two different Reed-Muller codes. 

\begin{figure}
\centering
\begin{minipage}{.45\textwidth}
  \centering
  \includegraphics[width=1\linewidth]{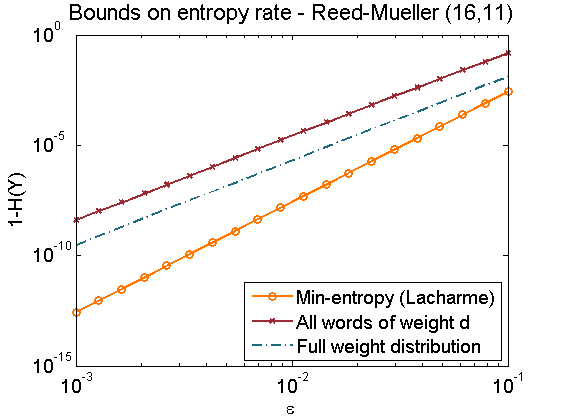}
  \caption{Bounds on the entropy of the output $\bar{Y}$ after the application of a post-processing built from a Reed-Muller [16,11] code.}
\label{figure: RM16}
\end{minipage}%
\hspace{.05\textwidth}
\begin{minipage}{.45\textwidth}
  \centering
  \includegraphics[width=1\linewidth]{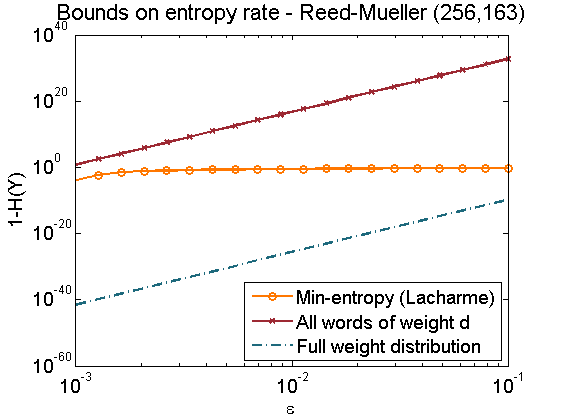}
  \caption{Bounds on the entropy of the output $\bar{Y}$ after the application of a post-processing built from a Reed-Muller [256,163] code.}
\label{figure: RM256}
\end{minipage}
\end{figure}

%
%
%\begin{center}
%\begin{figure}
%\includegraphics[scale=.6]{bounds_RM16.png}
%\caption{Bounds on the entropy of the output $\bar{Y}$ after the application of a post-processing built from a Reed-Muller [16,11] code.}\label{figure: RM16}
%\end{figure}
%\end{center}
%
%\begin{center}
%\begin{figure}
%\includegraphics[scale=.6]{bounds_RM256.png}
%\caption{Bounds on the entropy of the output $\bar{Y}$ after the application of a post-processing built from a Reed-Muller [256,163] code.}\label{figure: RM256}
%\end{figure}
%\end{center}

\section{Conclusions} % and further research

In this work we consider the known results on the upper bound on the output bias found in Theorems \ref{binary extractor thm} and \ref{thm infinity norm} and we introduce a new bound for the same processing considering the output as the whole vector of $k$ bits and computing the TVD instead of the minimal entropy. 
The known bound relies on the minimum distance of the code generated by the matrix. Using this we can choose an appropriate extractor by choosing a generator matrix whose code has a large enough minimum distance.
Using the new bound presented in Theorem \ref{theorem weight distribution}, we can choose the linear extractor by looking at the weight distribution of linear codes whose minimum distance is large enough to achieve the required minimal entropy.
\\
In section \ref{section: entropy} we used the known Theorem \ref{thm: Sason} to obtain a lower bound for the entropy rate of the output process. Experimental results show how in certain cases the knowledge of the entire weight distribution of the linear code associated to the post-processing is helpful to obtain good bounds on the output entropy.

\section{Acknowledgements}
These results come from the first author's MSc thesis and so he would like to thank his tutor (the third author) and his supervisor (the second author).

Part of this research was funded by the Autonomous Province of Trento, Call “Grandi Progetti 2012”, project “On silicon quantum optics for quantum computing and secure communications – SiQuro”.

\bibliographystyle{amsalpha}
\bibliography{RefsCGC}

\end{document}